\documentclass[journal]{IEEEtran}

\usepackage[cmex10]{amsmath}
\usepackage{amssymb}
\usepackage{amsthm}
\DeclareMathOperator*{\argmax}{arg\,max}  
\DeclareMathOperator*{\cov}{Cov} 
\DeclareMathOperator*{\var}{Var}	

\usepackage{dsfont} 

\newtheorem{theorem}{Theorem}
\newtheorem{lemma}[theorem]{Lemma}

\newtheorem{definition}{Definition}
\newtheorem{remark}{Remark} 
\newtheorem{assumption}{Assumption}
\newtheorem*{problemdefinition}{Problem Definition} 

\usepackage[ruled,vlined]{algorithm2e}
 
\usepackage{enumerate} 

\usepackage{graphicx}
\usepackage{caption}
\usepackage{subcaption}

\newcommand{\ttdist}{{\pi_{tt}}} 
 
\newcommand{\pardist}{{p_\theta}} 
\newcommand{\parmat}{{P_\theta}} 
\newcommand{\parstatdist}{{\pi_{\theta}}}
\newcommand{\seedatn}{{\mathcal{I}_n}}
\newcommand{\seed}{{\mathcal{I}}}

\newcommand{\hmmpost}{\pi_n^\theta}
\newcommand{\hmmprior}{\pi_0}
\ifCLASSINFOpdf
\else
\fi
\begin{document}
%
\title{Influence Maximization over Markovian Graphs: A Stochastic Optimization Approach}
%
%
%

\author{Buddhika~Nettasinghe,~\IEEEmembership{Student~Member,~IEEE}
        and~Vikram~Krishnamurthy,~\IEEEmembership{Fellow,~IEEE}
\thanks{Authors are with the Department
of Electrical and Computer Engineering, Cornell~University and, Cornell~Tech, 2 West Loop Rd, New~York, NY 10044, USA e-mail:  \{dwn26, vikramk\}@cornell.edu.}
}

\maketitle

\begin{abstract}
This paper considers the problem of randomized influence maximization over a Markovian graph process: given a fixed set of nodes whose connectivity graph is evolving as a Markov chain, estimate the probability distribution (over this fixed set of nodes) that samples a node which will initiate the largest information cascade (in expectation). Further, it is assumed that the sampling process affects the evolution of the graph i.e. the sampling distribution and the transition probability matrix are functionally dependent. In this setup, recursive stochastic optimization algorithms are presented to estimate the optimal sampling distribution for two cases: 1) transition probabilities of the graph are unknown but, the graph can be observed perfectly 2) transition probabilities of the graph are known but, the graph is observed in noise. These algorithms consist of a neighborhood size estimation algorithm combined with a variance reduction method, a Bayesian filter and a stochastic gradient algorithm. Convergence of the algorithms are established theoretically and, numerical results are provided to illustrate how the algorithms work. 
\end{abstract}


%
\IEEEpeerreviewmaketitle

\section{Introduction}
\label{sec:introduction}
{\it Influence maximization} refers to the problem of identifying the most influential node (or the set of nodes) in a network, which was first studied in the seminal paper \cite{kempe2003}. However, most work related to influence maximization so far has been limited by one or more of the following assumptions:
\begin{enumerate}
\item Deterministic network (with no random evolution)

\item Fully observed graph (instead of noisy observations of the graph)

\item Passive nodes (as opposed to active nodes that are responsive to the influence maximization process).
\end{enumerate}

This paper attempts to relax the above three assumptions. We develop stochastic optimization algorithms for influence maximization over a randomly evolving, partially observed network of active nodes (see Fig. \ref{fig:block_diagram_optimzation} for a schematic overview of our approach).

\begin{figure*}[h!]
	\centering
	\includegraphics[width=5.5in,height=2.75 in]{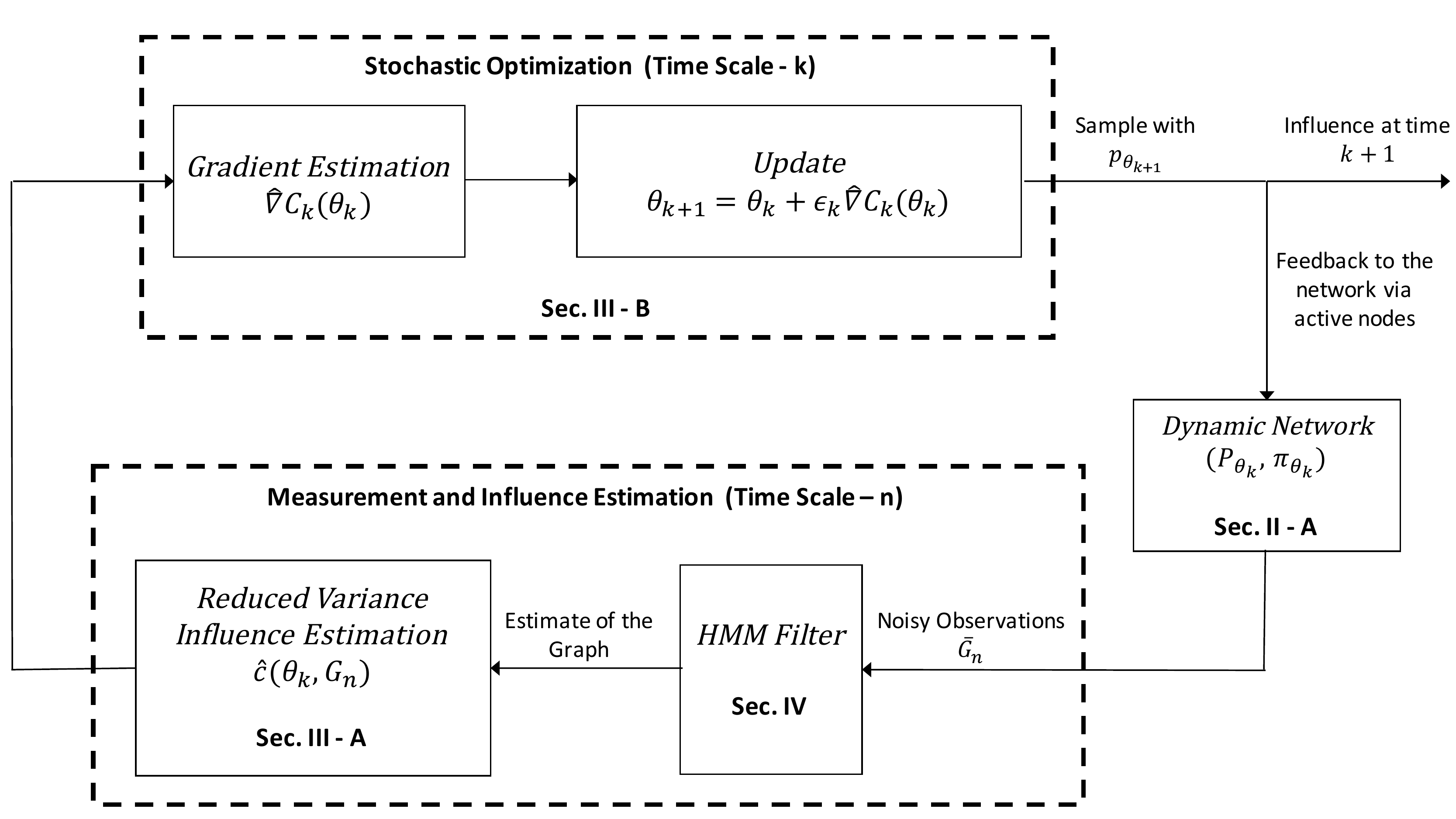}
	\caption{Schematic diagram of the proposed stochastic optimization algorithm for influence maximization over a partially observed dynamic network of active nodes, showing how the algorithmic sub-components are integrated to form the closed loop (with feedback from the sampling process) stochastic optimization algorithm and their organization in the paper.}
	\label{fig:block_diagram_optimzation}
\end{figure*}

To understand the motivation behind this problem, consider a social network graph where nodes represent individuals and the directed edges represent connectivity between them. Assume that this graph evolves in a Markovian manner\footnote{Dynamics of social networks (such as seasonal variations in friendship networks) can naturally be modeled as Markov processes. Another example would be a vehicular network where, the inter-vehicle communication/connectivity graph has a Markovian evolution due to their movements. Refer \cite{hamdi2014_tracking} for an example in the context of social networks.} with time. Further, each individual can pass/receive messages (also called infections depending on the context) from their neighbors by communicating over the directed edges of the graph. Communication over these edges will incur time delays that are independently and identically distributed (across the edges) according to a known distribution. An influence maximizer wants to periodically convey messages (e.g. viral marketing) that expire in a certain time window to the nodes in this evolving social network. Conveying each message to the nodes in the social network is achieved by sampling a node (called seed node) from the graph according to a probability distribution and, giving the message to that seed node. Then, the seed node initiates an information cascade by transmitting the message to its neighbors with random delays. The nodes that receive the message from their neighbors will continue to follow the same steps, until the message expires. It is assumed that the graph remains same throughout the diffusion of one message i.e. graph evolves on a slower time scale compared to the expiration time of a message. Further, we allow the nodes of the social network to be active: nodes are aware of the sampling distribution and, respond by modifying the transition probabilities of the graph according to that distribution (for example, due to the incentive that they receive for being the most influential node\footnote{Another example for active nodes is a network of computers that are adaptively modifying their connectivity network, depending on how vulnerable each computer is to a virus attack.}). This makes the transition probability matrix functionally dependent on the sampling distribution. In this setting, the goal of the influence maximizer is to compute the sampling distribution which maximizes the expected total number of nodes that are infected (size of the information cascade) before a message expires (considering the randomness of sampling the nodes, message diffusion process as well as the graph evolution). This motivates us to pursue the aim of this paper, which is to devise a method for the influence maximizer to estimate the optimal sampling distribution recursively, with each message that is distributed. 

\textbf{The main results} of this paper are  two recursive stochastic gradient algorithms, for the influence maximizer to recursively estimate (track) the optimal sampling distribution for the following cases:
\begin{enumerate}
\item \label{contribution:case1}influence maximizer does not know the transition probability matrix but, has perfect (non-noisy) observations of the sample path of the graph.

\item \label{contribution:case2}influence maximizer knows the transition probability matrix but, has only partial (noisy) observations of the sample path of the graph  evolution.
\end{enumerate}
The key components of the above two algorithms (illustrated in Fig. \ref{fig:block_diagram_optimzation}) include the following.
\begin{itemize}
	\item {\bf Reduced variance neighborhood size estimation algorithm}: Influence maximization problems involve estimating the influence of nodes, which can be posed as a problem of estimating the (expected) sizes of node neighborhoods. For this, we use a stochastic simulation based neighborhood size estimation algorithm (which utilizes a modified Dijkstra's algorithm, combined with an exponential random variable assignment process), coupled with a variance reduction approach. It is shown that this reduced variance method improves the convergence of the proposed algorithms when tracking the optimal influence in a time evolving system. 
	
	\item {\bf Stochastic optimization with delayed observations of the graph process:} The observations of the graph sample path in the two algorithms (in main contributions) are not assumed to be available in real time. Instead, it is sufficient if the sample paths of finite lengths become available as batches of data with some time delay\footnote{In most real world applications, one can only trace back the evolution of a social network over a period of time (length of the finite sample path), instead of monitoring it in real time. e.g. how the graph has evolved over the month of January becomes available to the influence maximizer only at the end of February due to delays in obtaining data.}.These finite length graph sample paths are used in a stochastic optimization method that is based on the simultaneous perturbation stochastic approximation (SPSA) method, coupled with a finite sample path gradient estimation method for Markov processes. The proposed algorithms are applicable even for the more general case where, the system model (state space of the Markovian graph process, the functional dependency between the sampling distribution and transition probabilities, etc) is varying on a slower (compared to the stochastic optimization algorithm) time scale.

	\item {\bf Bayesian filter for noisy graph sample paths: }In the algorithm for the second case (in main contributions), the sample paths are assumed to be observed in noise. In this case, a Bayesian filter is utilized to estimate the underlying state of the graph using the noisy sample path as the input. The estimates computed by the Bayesian filter are then utilized in the stochastic optimization process while preserving the (weak) convergence of the stochastic approximation algorithm. 
\end{itemize}

\vspace{0.2 cm}
{\bf Related Work:} The influence maximization problem was first posed in \cite{kempe2003} as a combinatorial optimization problem for two widely accepted models of information spreading in social networks: independent cascade model and the linear threshold model. \cite{kempe2003} shows that solving this problem is NP-hard for both of these models and, utilizes a greedy (submodular function) maximization approach to devise algorithms with a $1 - \frac{1}{e}$ approximation guarantee. Since then, this problem and its variants have been widely studied using different techniques and models. \cite{bharathi2007, borodin2010threshold, carnes2007, chen2011} studies the problem in a competitive/adversarial settings with multiple influence maximizers and provide equilibrium results and approximation algorithms. \cite{seeman2013,horel2015scalable} considers the case where all the nodes of the graph are not initially accessible to the influence maximizer and, proposes a multistage stochastic optimization method that harvests the power of a phenomenon called friendship paradox \cite{feld1991,lattanzi2015}.  Further, \cite{chen2009,chen2010scalable} provide heuristic algorithms for the influence maximization problem on the independent cascade model, which are more efficient compared to the originally proposed algorithms in \cite{kempe2003}. Our work is motivated by the problems studied in \cite{du2017,zhuang2013influence}. \cite{du2017} points out that estimating the expected cascade size under the independent cascade model can be posed as a neighborhood size estimation problem on a graph and, utilizes a size estimation framework proposed in \cite{cohen1997} (and used previously in \cite{chen2010}) to obtain an unbiased estimate of this quantity. \cite{zhuang2013influence} highlights that the study of the influence maximization problem has mostly been limited to the context of static graphs and, proposes a random probing method for the case where the graph may evolve randomly. Motivated by these observations, we focus on a continuous time variant \cite{du2017,rodriguez2011} of the independent cascade model and, allow the underlying social network graph to evolve slowly as a Markov process. The solution approaches proposed in this paper belongs to the class of recursive stochastic approximation methods. These methods have been utilized previously to solve many problems in the field of multi-agent networks \cite{krishnamurthy2014interactive, sayed2014adaptation,gharehshiran2013}. 

\vspace{0.2 cm}
{\bf Organization:} Sec. \ref{sec:diffusion_model_and_influence} presents the network model, related definitions of influence on graphs, the definition of the main problem and finally, a discussion of some practical details. Sec. \ref{sec:stochastic_optimization_perfectly_observed_graph} presents the recursive stochastic optimization algorithm (along with convergence theorems) to solve the main problem for the case of fully observed graph with unknown transition probabilities (case \ref{contribution:case1} of the main contributions). Sec. \ref{sec:stochastic_optimization_partially_observed_graph} extends to the case of partially observed graph with known transition probabilities (case \ref{contribution:case2} of the main contributions). Sec. \ref{sec:numerical_results} provides numerical results to illustrate the algorithms presented.

\section{Diffusion Model and the Problem of Randomized Influence Maximization }
\label{sec:diffusion_model_and_influence}

This section describes Markovian graph process, how information spreads in the graph (information diffusion model) and, provides the definition of the main problem. Further, motivation for the problem, supported by work in recent literature, is provided to highlight some practical details. 

\subsection{Markovian Graph Process, Frequency of the Messages and the Independent Cascade (IC) Diffusion Model}
\label{subsec:independent_cascade}
 {\bf Markovian Graph Process:} The social network graph at discrete time instants $n = 0,1,\dots$ is modeled as a directed graph $G_n = (V, E_n)$, consisting of a fixed set of individuals $V$, connected by the set of directed edges $E_n$. The graph evolves as a Markov process $\{G_n = (V, E_n)\}_{n \geq 0}$ with a finite state space $\mathcal{G}$, and a parameterized regular transition matrix $\parmat$ with a unique stationary distribution $\parstatdist$ (where, $\theta$ denotes the parameter vector which lies in a compact subset of an $M$-dimensional Euclidean space). Henceforth, $n$ is used to denote the discrete time scale on which the Markovian graph process evolves.

 \vspace{0.25cm} 
 {\bf Frequency of the Messages:} An influence maximizer distributes messages to the nodes in this evolving network. We assume that the messages are distributed periodically\footnote{The assumption of periodic messages is not required for the problem considered in this paper and the proposed algorithms. It suffices if the messages are distributed with some minimum time gap between them (on the time scale $n$).} at time instants $n= kN$ where, the positive integer $N$ denotes the period.
  
 \vspace{0.25cm} 
 {\bf Observations of the Finite Sample Paths:} When the influence maximizer distributes the $(k+1)$\textsuperscript{th} message at time ${n = (k+1)N}$, only the finite sample path of the Markovian graph process $\{G_n\}_{kN \leq n < kN+\bar{N}}$ for some fixed $\bar{N} \in \mathbb{N}$ with $\bar{N} < N$, is visible to the influence maximizer.

\vspace{0.25cm} 
 {\bf Information Diffusion model:} As explained in the example in Sec. \ref{sec:introduction}, nodes pass messages they receive to their neighbors with random delays. This method of information spreading in graphs is formalized by independent cascade (IC) model of information diffusion. Various versions of this IC model have been studied in literature. We use a slightly different version of the IC model utilized in \cite{du2017,rodriguez2011} and, it is as follows briefly. When the influence maximizer gives a message to a set of seed nodes\footnote{We consider the case where the influence maximizer selects a set of seed nodes instead of one seed node (as in the motivating example in Sec. \ref{sec:introduction}) to keep the definitions of this section more general.} $\seedatn \subseteq V$ at time $n$, a time variable $t  \in \mathbb{R}_{\geq0}$ is initialized at $t=0$. Here, $t$ is the continuous time scale on which the diffusion of the message takes place and, is different from the discrete time scale $n\in \mathbb{Z}_{\geq0}$ on which the graph evolves (and the  messages are distributed periodically). Further, the time scale $t$ is nested in the time scale $n$: $t$ is set to $t=0$ with each new message distributed by the influence maximizer. The set of seed nodes $\seedatn$ (that receive the message from influence maximizer at $n$ and $t = 0$) transmits the message through edges attached to them in graph $G_n = (V,E_n)$. Each edge $(j,i) \in E_n$ induces a random delay distributed according to probability density function $ \ttdist(\tau_{ji})$ (called transmission time distribution). Further, each edge transmits the message only once.  Only the neighbor which infects a node first will be considered as the true parent node (considering the infection propagation) of the infected node. This process continues until the message expires at $t = T$ (henceforth referred to as the message expiration time) and, the diffusion process stops at that time. It is assumed that the discrete time scale $n$ on which the graph evolves is slower than $T$ i.e. the graph will remain the same at least for $t \in [0,T]$. Then, the same message spreading process will take place when the next message is distributed.
 
 For any realization of this random message spreading process (taking place in $t \in [0,T]$), the subgraph of $G_n$ which is induced by the set of edges through which the message propagated, constitutes a Directed Acyclic Graph (DAG). Further, due to the DAG induced by the propagation of a message, the infection time $t_i$ of each node $i \in V$ satisfies the \textit{shortest path property}: conditional on a graph $G_n = (V, E_n)$, and a set of pairwise transmission times $\{\tau_{ji}\}_{(ji) \in E_n}$, the infection time $t_i$ of $i \in V$ is given by, 
\begin{equation}
\label{eq:sp_property}
t_i = g_i(\{\tau_{ji}\}_{(ji) \in E_n}\vert\, \seed)
\end{equation}
where, $g_i(\{\tau_{ji}\}_{(ji) \in E_n}\vert\, \seed)$ denotes the length of the shortest path from $\seed \subseteq V$ to $i \in V$, with edge lengths $\{\tau_{ji}\}_{(ji) \in E_n}$.

\subsection{Influence of a Set of Nodes Conditional on a Graph}
\label{subsec:estimating_the_influence_of_a_set}

We use the following definition (from \cite{du2017,rodriguez2012}) of influence of a set of nodes $\seed \subseteq V$, on a graph  $G = (V,E)$, for the diffusion model introduced in Sec. \ref{subsec:independent_cascade}. Further, $\mathbb{E}_{\ttdist}\{\cdot\}$ is used to denote expectation over a set of pairwise transmission times (associated with each edge) sampled independently from $\ttdist$.
\begin{definition}
\label{defn:influence_on_G}
The Influence, $\sigma_G(\seed,T)$, of a set of nodes $\seed$, given a graph ${G = (V,E)}$ and a fixed time window $T$ is,
	\begin{align}
	\sigma_G(\seed,T) &= \mathbb{E}_{\ttdist} \Big\{\sum_{i \in V} \mathds{1}_{\{t_i \leq T\}}\bigr\vert \seed,G   \Big\}	\label{eq:influence_on_G_1}\\ 
	&= \sum_{i \in V} \mathbb{P}_{\ttdist} \big\{{t_i \leq T}\bigr\vert \seed,G   \big\}. 
	\label{eq:influence_on_G_2}
	\end{align}
\end{definition}
In (\ref{eq:influence_on_G_1}), the influence $\sigma_G(\seed,T)$ of the set of nodes ${\seed \subset V}$ is the expected number of nodes infected within time $T$, by the diffusion process (characterized by the distribution of transmission delays $\ttdist$) on the graph $G$, that started with the set of seed nodes $\seed$. 

Note that the set of infection times $
\{t_i\}_{i \in V}$, in (\ref{eq:influence_on_G_1}) are dependent random variables. Therefore, obtaining closed form expressions for marginal cumulative distributions in (\ref{eq:influence_on_G_2}) involves computing $\vert V \vert -1$ dimensional integral which is not possible (in closed form) for many general forms of the transmission delay distribution. Further, numerical evaluation of these marginal distributions is also not feasible since it will involve discretizing the domain $[0,\infty)$ (refer \cite{du2017} for a more detailed description about the computational infeasibility of the calculation of the expected value in (\ref{eq:influence_1})).

{\bf Why randomized selection of seed nodes?} Assume that, at time instant $n = kN$, an influence maximizer needs to find a seed node $i^* \in V$ that maximizes $\sigma_{G_n}(\{i\}, T)$ (in order to distribute the $k$\textsuperscript{th} message, as explained in Sec. \ref{subsec:independent_cascade}). To achieve this, the influence maximizer needs to know the graph $G_n$, calculate the influence $\sigma_{G_n}(\{i\}, T)$ for each $i \in V$ and then, locate the $i^* \in V$ that has the largest influence. However, performing all these steps for each message is not feasible from a practical perspective. Especially, monitoring the graph in real time is practically difficult. Hence, a natural alternative is to use a randomized seed selection method. We consider a case where the seed node $i \in V$ is sampled from a parameterized probability distribution $p_{\theta} \in \Delta(V) $ such that, the expected size of the cascade initiated by the sampled node $i$ is largest when $\theta = \theta^*$. This random seed selection approach leads to the formal definition of the main problem which we call the randomized influence maximization. 

\subsection{Randomized Influence Maximization over a Markovian Graph Process: Problem Definition}
\label{subsec:problem_formulation}

We define the influence of a parameterized probability distribution $\pardist,\, \theta \in \mathbb{R}^M$, on a Markovian graph process as below (henceforth, we will use $\sigma_G(i,T)$ to denote $\sigma_G(\{i\},T)$ with a slight notational misuse).
\begin{definition}
	\label{defn:influence_function}
	The Influence, $C(\theta)$, of probability distribution $\pardist \in \Delta(V)$, on a Markovian graph process ${\{G_n = (V, E_n)\}_{n\geq 0}}$ with a finite state space $\mathcal{G}$, a regular transition matrix $\parmat$ and, a unique stationary distribution $\parstatdist \in \Delta(\mathcal{G})$ is,
	\begin{equation}
		\label{eq:influence_1}
		C(\theta) = \mathbb{E}_{G\sim \parstatdist}\{c(\theta,G)\}
	\end{equation} where, 
	\label{defn:influence}
	\begin{equation}
		\label{eq:influence_2}
		c(\theta,G) = \mathbb{E}_{i\sim \pardist} \{\sigma_G(i,T)\}.
	\end{equation}
\end{definition}
Eq. (\ref{eq:influence_2}) averages influence $\sigma_G(i,T)$ using the sampling distribution $\pardist$, to obtain $c(\theta, G)$, which is the influence of the sampling distribution $\pardist$ conditional on the graph $G$. Then, (\ref{eq:influence_1}) averages $c(\theta, G)$ using the unique stationary distribution $\parstatdist$ of the graph process in order to obtain $C(\theta)$, which is the influence of  the sampling distribution $\pardist$ over the Markovian graph process. We will refer to $c(\theta, G_n)$ as the conditional (on graph $G_n$) influence function (at time $n$) and, $C(\theta)$ as the influence function, of the sampling distribution $\pardist$. 
\begin{remark}
	\normalfont
	\label{remark:functional_dependency}
	The sampling distribution $\pardist$ is treated as a function of $\theta$, which also parameterizes the transition matrix $\parmat$ of the graph process. This functional dependency models the feedback (via the active nodes) from the sampling of nodes (by the influence maximizer) to the evolution of the graph, as indicated by the feedback loop in Fig. \ref{fig:block_diagram_optimzation} (and also discussed in an example setting in Sec. \ref{sec:introduction}). 
\end{remark}

In this context, the main problem studied in this paper can be defined as follows. 

\vspace{0.25cm}
\begin{problemdefinition}

	 Randomized influence maximization over a Markovian graph process ${\{G_n = (V, E_n)\}_{n\geq 0}}$ with a finite state space $\mathcal{G}$, a regular transition matrix $P_{\theta_{k'}}$ and, a unique stationary distribution $\pi_{\theta_{k'}}$ aims to recursively estimate the time evolving optima,
\begin{equation}
	\label{eq:optimization_problem}
	\theta^*_{k'} \in \argmax_{\theta \in \mathbb{R}^M }\, C_{k'}(\theta)
\end{equation}
where, $C_{k'}(\theta)$ is the influence function (Definition \ref{defn:influence_function}) that is evolving on the slower time scale $k' \in \mathbb{Z}_{\geq 0}$ (compared to the message period $N$ over the time scale $n$) .
\end{problemdefinition}

\begin{remark}
	\normalfont
	The reason for allowing the influence function $C_{k'}(\theta)$ (and therefore, the solution $\theta^*_{k'}$) in (\ref{eq:optimization_problem}) to evolve over the slow time scale $k'$ is because the functional dependency of the sampling distribution and the transition probability matrix (which gives how the graph evolution depends on sampling process) may change over time. Further, the state space $\mathcal{G}$ of the graph process may also evolve over time. Such changes (with time) in the system model are encapsulated by modeling the influence function as a time evolving quantity.  However, to keep the notation manageable, we assume that the influence function $C_{k'}(\theta)$ does not evolve over time in the subsequent sections i.e. it is assumed that 
	\begin{equation}
	C_{k'}(\theta) = C(\theta), \forall\,k' \in \mathbb{Z}_{\geq0}.
	\end{equation}
	 This assumption is used to keep the notation manageable and, can be removed without affecting the main algorithms presented in this paper. Further, we assume that the $C(\theta)$ has a Lipschitz continuous derivative.  
\end{remark}

\subsection{Discussion about Key Aspects of the System Model}
\label{subsec:practical_example}
{\bf Networks as Markovian Graphs:} We assumed that the graph evolution is Markovian. In a similar context to ours, \cite{clementi2009information} states that ``Markovian evolving graphs are a natural and very general class of models for evolving graphs" and, studies the information spreading protocols on them during the stationary phase. Further, \cite{clementi2016rumor} considers information broadcasting methods on Markovian graph processes since they are ``general enough for allowing us to model basically any kind of network evolution''.

\vspace{0.25cm}
{\bf Functional Dependency of the Sampling Process and Graph Evolution:}  \cite{clementi2009broadcasting} considers a \textit{weakly adversarial random broadcasting network}: a randomly evolving broadcast network whose state at the next time instant is sampled from a distribution that minimizes the probability of successful communications. Analogous to this, the functional dependency in our model may represent how the network evolves adversely to the influence maximization process (or some other underlying network dynamic which is responsive to the influence maximization). The influence maximizer need not be aware of such dependencies to apply the algorithms that will be presented in the next sections.

\section{Stochastic Optimization Method: Perfectly Observed Graph Process with Unknown Dynamics}
\label{sec:stochastic_optimization_perfectly_observed_graph}
In this section, we propose a stochastic optimization method for	 the influence maximizer to recursively estimate the solution of the optimization problem in (\ref{eq:optimization_problem}) under the Assumption \ref{assumption:fully_observed_graph} stated below. 

\begin{assumption}
	\label{assumption:fully_observed_graph}
	The influence maximizer can fully observe the sample paths of the Markovian graph process, but does not know the transition probabilities with which it evolves. 
\end{assumption}

The schematic overview of the approach for solving (\ref{eq:optimization_problem}) is shown in Fig. \ref{fig:block_diagram_optimzation} (where, the HMM filter is not needed in this section due to the Assumption \ref{assumption:fully_observed_graph}). In the next two subsections, the conditional influence estimation algorithm and the stochastic optimization algorithm will be presented.

\subsection{Estimating the Conditional Influence Function using Cohen's Algorithm}
\label{subsec:estimating_the_instantaneous_cost}
The exact computation of the node influence $\sigma_G(i, T)$ in closed form or estimating it with a naive sampling approach is computationally infeasible (as explained in Sec. \ref{subsec:estimating_the_influence_of_a_set}). As a solution, \cite{du2017} shows that the shortest path property of the IC model (explained in Sec. \ref{subsec:independent_cascade}) can be used to convert (\ref{eq:influence_on_G_1}) into an expression involving a set of independent random variables as follows:

\begin{equation}
\label{eq:influence_on_G_sp}
\sigma_G(\seed,T) = \mathbb{E}_{\ttdist} \Big\{\sum_{i \in V} \mathds{1}_{\{g_i(\{\tau_{ji}\}_{(ji) \in E}\vert \seed) \leq T\}}\Big\vert \seed,G \Big\}
\end{equation} 
where, $g_i(\{\tau_{ji}\}_{(ji) \in E}\vert \seed)$ is the shortest path as defined previously in (\ref{eq:sp_property}). Further, note from (\ref{eq:influence_on_G_sp}) that influence of the set $\seed$, $\sigma_G(\seed,T)$ is the expected $T$-distance neighborhood (expected number of nodes within $T$ distance from the seed nodes) i.e.
\begin{equation}
\label{eq:influence_of_a_node_on_a_graph}
\sigma_G(\seed,T) = \mathbb{E}_{\ttdist} \big\{\big[\vert \mathcal{N}_G(\seed,T)\vert\big]\big\vert \seed,G\big\}
\end{equation} where,
\begin{equation}
\mathcal{N}_G(\seed,T) = \{i\in V: g_i(\{\tau_{ji}\}_{(ji) \in E}\vert \seed) \leq T\}.
\end{equation}
Hence, we only need a neighborhood size estimation algorithm and samples from $\ttdist$ to estimate the influence $\sigma_G(\seed,T)$ of the set $\seed \subset V$.

\subsubsection{Cohen's Algorithm} Based on (\ref{eq:influence_of_a_node_on_a_graph}), \cite{du2017} utilizes a neighborhood size estimation algorithm proposed in \cite{cohen1997} in order to obtain an unbiased estimate of influence $\sigma_G(\seed,T)$.  This algorithm is henceforth referred to as Cohen's algorithm. The main idea behind the Cohen's algorithm is the fact that the minimum of a finite set of unit mean exponential random variables is an exponential random variable with an exponent term equal to the total number random variables in the set.  Hence, for a given graph ${G = (V,E)}$ and a transmission delay set $L_u = {\{\tau_{ji}\}^u_{(ji) \in E}}$ (where $u$ denotes the index of the set of transmission delays), this algorithm assigns $m$ number of exponential random variable sets ${l_{u,j} = \{r_v^{u,j} \sim \exp(1): v\in V\}}$, where, $j = 1,\dots,m$. Then, a modified Dijkstra's algorithm (refer \cite{cohen1997,du2017} for a detailed description of the steps of this algorithm) finds the smallest exponential random variable $\bar{r}^{\seed}_{u,j}$ within $T-$distance from the set $\seed$, for each $l_{u,j}$ with respect to the transmission time set $L_u$. Then, \cite{cohen1997} shows that $\mathbb{E}\big\{\frac{m-1}{\sum_{j= 1}^{m}\bar{r}^{\seed}_{u,j}}\big\vert L_u\big\}$ is an unbiased estimate of $\mathcal{N}_G(\seed,T)$ conditional on $L_u$. Further, this Cohen's algorithm for estimating $\sigma_G(\seed,T)$ has a lower computational complexity which is near linear in the size of network size, compared to the computational complexity of a naive simulation approach (repeated calling of shortest path algorithm and averaging) to estimate $\sigma_G(\seed,T)$ \cite{du2017}.  

\subsubsection{Reduced Variance Estimation of Influence using Cohen's algorithm} We propose Algorithm \ref{alg:instantaneous_cost_estimation} in order to estimate the conditional influence function $c(\theta, G)$.  First four steps of Algorithm \ref{alg:instantaneous_cost_estimation} are based on a reduced variance version of the Cohen's algorithm employed in \cite{du2017} called CONTINEST.  The unbiasedness and the reduced (compared to the algorithm used in \cite{du2017})  variance of the estimates obtained using Algorithm \ref{alg:instantaneous_cost_estimation} are established in Theorem \ref{th:unbiasedness}. 

\begin{algorithm}
	\caption{Conditional influence estimation algorithm}
	\label{alg:instantaneous_cost_estimation}
	\DontPrintSemicolon 
	\KwIn{Sampling distribution $\pardist$, Cumulaive Distribution $F_{tt}(\cdot)$ of the transmission time distribution $\ttdist$, Directed graph $G = (V, E)$}
	\KwOut{Estimate of conditional influence function: $\hat{c}(\theta, G)$}
	\vspace{0.3cm}
	For all $v \in V$, execute the following steps simultaneously.
	\begin{enumerate}	
		\item Generate $s/2$ sets of uniform random variables: 
		\begin{equation*}
		\{\{U^u_{ji}\}_{(ji) \in E}: u = 1,2,\dots, s/2\} 
		\end{equation*}

		\item For each $u = 1,\cdots, s/2$, generate a correlated pair of random transmission time sets as follows: 
		\begin{align}
			L_u &= \{F_{tt}^{-1}(U^u_{ji})\}_{(ji) \in E} \label{eq:correlated_transmission_time_sets_1}\\
			L_{s/2+u} &= \{F_{tt}^{-1}(1 - U^u_{ji})\}_{(ji) \in E} \label{eq:correlated_transmission_time_sets_2}
		\end{align}

		\item For each set $L_u$ where $u = 1,\cdots, s$, assign $m$ sets of independent exponential random variables: ${l_{u,j} = \{r_v^{u,j} \sim \exp(1): v\in V\}},\, j = 1,\dots,m$.
		
		\item Compute the minimum exponential random variable $\bar{r}^{v}_{u,j}$ that is within $T$-distance from $v$ using the modified Dijkstra's algorithm, for each $l_{u,j}$. Calculate,
		\begin{equation}
			\label{eq:cohen_estimate}
			X(v,G) = \frac{1}{s} \sum_{u = 1}^{s} \frac{m-1}{\sum_{j= 1}^{m}\bar{r}^{v}_{u,j}}
		\end{equation}
		
		\item Compute,
		\begin{equation}
			\label{eq:instantaneous_cost_estimate}
			\hat{c}(\theta, G) =  \sum_{v \in V}\pardist(v) X(v,G).
		\end{equation}
	\end{enumerate}
\end{algorithm}

\begin{theorem}
\label{th:unbiasedness}
Consider a graph $G = (V,E)$.
\begin{enumerate}[I.]
	\item Given $v \in V$, $X(v,G)$ in (\ref{eq:cohen_estimate}) is an unbiased estimate of node influence $\sigma_G(v,T)$, with a variance 
	\begin{multline}
	\label{eq:reduced_variance_bound}
	Var(X(v,G)) \\
	\hspace{0.25cm}\leq \frac{1}{s}\bigg({\frac{\sigma_G(v,T)^2}{m-2} +  \frac{(m-1)Var(\vert\mathcal{N}_G(v,T)\vert)}{m-2}}\bigg)
	\end{multline}
	\label{th:unbiasedness_part1}

	\item The output $\hat{c}(\theta, G)$ of Algorithm \ref{alg:instantaneous_cost_estimation} is an unbiased estimate of the conditional influence function $c(\theta, G)$ defined in (\ref{eq:influence_2}). Variance of $\hat{c}(\theta, G)$ is bounded above by the variance in the case where $L_u$ and $L_{s/2+u}$ (in Step 2 of Algorithm \ref{alg:instantaneous_cost_estimation}) are independently generated.
	\label{th:unbiasedness_part2}
	
\end{enumerate}
\end{theorem}

\begin{proof}
See Appendix \ref{appen:unbiasedness}.
\end{proof}

Theorem \ref{th:unbiasedness} shows that the estimate $X(v,G)$ computed in Algorithm \ref{alg:instantaneous_cost_estimation} has a smaller variance compared to the estimate computed by the Cohen's algorithm based influence estimation method (named CONTINEST) used in \cite{du2017} which has a variance of $\frac{1}{s}\big({\frac{\sigma_G(v,T)^2}{m-2} +  \frac{(m-1)Var(\vert\mathcal{N}_G(v,T)\vert)}{m-2}}\big)$. The reason for this reduced variance is the correlation created by using the same set of uniform random numbers (indexed by $u$) to generate a pair of transmission time sets ($L_u$ and $L_{s/2+u}$). Due to this use of same random number number for multiple realizations, this method is referred to as the method of common random numbers \cite{ross2013simulation}. This reduced variance in the estimates $X(v,G)$ results in a reduced variance in the estimate $\hat{c}(\theta, G)$ of the conditional influence ${c}(\theta, G)$. 

\subsection{Stochastic Optimization Algorithm}
\label{subsec:stoch_optimization}

\begin{algorithm}
	\caption{SPSA  based algorithm to estimate $\theta^*$}
	\label{alg:stochastic_optimization}
	\DontPrintSemicolon 
	\KwIn{Initial parameterization $\theta_0$, Transmission time distribution $\ttdist$, Observations of the graph process $\{G_n\}_{n\geq0}$}
	\KwOut{Estimate of the (locally) optimal solution $\theta^*$ of $C(\theta)$}
	\vspace{0.3cm}
For $k = 0,1,\dots$, execute the following steps.			
	\begin{enumerate}
		\item Simulate the $M$ dimensional vector $d_k$ with random elements 
		$$
		d_k(i)= 
		\begin{cases}
		+1 \text{ \,with probability } 0.5\\
		-1 \text{ \,with probability } 0.5.
		\end{cases}
		$$
		
		\item Set $\theta = \theta_k +  \Delta d_k$ where, $\Delta > 0$.
		
		\item Sample a node from the network using $p_{\theta}$ and, distribute the $2k$\textsuperscript{th} message with the sampled node as seed. 
		
		\item Obtain $\hat{c}(\theta, G_n)$ for $n = 2kN, 2kN+1, \dots, 2kN +\bar{N}-1$ using Algorithm \ref{alg:instantaneous_cost_estimation} and, calculate
		\begin{equation}
		\label{eq:cost_estimate_plus}
			\hat{C}_k(\theta_k +  \Delta d_k ) = \frac{1}{\bar{N}}\sum_{n = 2kN}^{2kN +\bar{N}-1} \hat{c}(\theta,G_n).
		\end{equation}		
		
		\item  Set $\theta = \theta_k  -  \Delta d_k$. Sample a node from the network using $p_{\theta}$ and, distribute the $2k + 1$\textsuperscript{th} message with the sampled node as seed. 

		\item Obtain $\hat{c}(\theta, G_n)$ for $n = (2k+1)N, (2k+1)N+1, \dots, (2k+1)N + \bar{N}-1$ using Algorithm \ref{alg:instantaneous_cost_estimation} and, calculate
		\begin{equation}
		\nonumber
		\label{eq:cost_estimate_minus}
		\hat{C}_k(\theta_k -  \Delta d_k ) = \frac{1}{\bar{N}}\sum_{n = (2k+1)N}^{(2k+1)N +\bar{N}-1} \hat{c}(\theta,G_n).
		\end{equation}

		\item Obtain the gradient estimate,
		\begin{equation}
		\label{eq:gradient_estimate}
		\hat{\nabla}C_k(\theta_k) = \frac{\hat{C}_k(\theta_k + \Delta d_k) - \hat{C}_k(\theta_k - \Delta d_k)}{2\Delta }d_k.
		\end{equation}
		
		\item Update sampling distribution parameter $\theta_k$ via stochastic gradient algorithm
		\begin{equation}
		\label{eq:model_update}
		\theta_{k+1} = \theta_{k} + \epsilon \hat{\nabla}C_k(\theta_k)
		\end{equation} where, $\epsilon  > 0.$
	\end{enumerate}
\end{algorithm}

We propose Algorithm \ref{alg:stochastic_optimization}  for solving the optimization problem (\ref{eq:optimization_problem}), utilizing the conditional influence estimates obtained via Algorithm \ref{alg:instantaneous_cost_estimation}. Algorithm \ref{alg:stochastic_optimization} is based on the Simultaneous Perturbation Stochastic Approximation (SPSA) algorithm (see \cite{spall1992,spall2003} for details). In general, SPSA algorithm utilizes  a finite difference estimate $\hat{\nabla}C_k(\theta_k)$, of the gradient of the function $C(\theta)$ at the point $\theta_k$ in the ($k$\textsuperscript{th} iteration of the) recursion,
\begin{equation}
\label{eq:gradient_algorithm}
\theta_{k+1} = \theta_{k} + \epsilon_k\hat{\nabla}C_k(\theta_k).
\end{equation} 

In the $k$\textsuperscript{th} iteration of the Algorithm \ref{alg:stochastic_optimization}, the influence maximizer  passes a message to the network using a seed node sampled from the distribution $\pardist$ where, $\theta = \theta_k +  \Delta d_k$ (step 3). Sampling with $\pardist$ causes the transition matrix to be become $\parmat$ (recall Remark \ref{remark:functional_dependency}). Then, in step 4, (\ref{eq:cost_estimate_plus}) averages the conditional influence estimates ${\hat{c}(\theta, G_n)}$ over $\bar{N}$ consecutive time instants  (where, $\bar{N}$ is the length of the available sample path as defined in Sec. \ref{subsec:independent_cascade}) to obtain ${\hat{C}_k(\theta_k +  \Delta d_k )}$, which is an asymptotically convergent estimate (by the law of large numbers for Markov Chains \cite{durrett2010_probability,norris1998markov}) of  ${C}(\theta_k +  \Delta d_k )$. Similarly, steps 5 and 6 obtain ${\hat{C}_k(\theta_k -  \Delta d_k )}$, which is an estimate of ${C}(\theta_k -  \Delta d_k )$. Using these influence function estimates,  (\ref{eq:gradient_estimate}) computes the finite difference gradient estimate $\hat{\nabla}C_k(\theta_k)$ in step 7. Finally, step 8 updates the $M$-dimensional parameter vector using the gradient estimate computed in (\ref{eq:gradient_estimate}). Some remarks about this algorithm are as follows. 

\begin{remark}
\normalfont
Algorithm \ref{alg:stochastic_optimization} operates in two nested time scales which are as follows (from the fastest to the slowest):
\begin{enumerate}
	\item $t\,$- continuous time scale on which the information diffusion takes place in a given realization of the graph.
	
	\item $n\,$- discrete time scale on which the graph evolves. 
\end{enumerate} Further, updating the parameter vector takes place periodically over the scale $n$, with a period of $2N$ (where, $N$ is the time duration between two messages as defined in Sec. \ref{subsec:independent_cascade}).
\end{remark}

\begin{remark}
\normalfont
	Note that all elements of the parameter vector are simultaneously perturbed in the SPSA based approach. Therefore, the parameter vector is updated once every two messages. This is in contrast to other finite difference methods such as Kiefer-Wolfowitz method, which requires $2M$ number of messages (where, $M$ is the dimension of the parameter vector as defined in Sec. \ref{subsec:independent_cascade}) for each update. 
\end{remark}

Next, we establish the convergence of Algorithm $\ref{alg:stochastic_optimization}$ using standard results which gives sufficient conditions for the convergence of recursive stochastic gradient algorithms (for details, see \cite{spall2003, kushner2003, krishnamurthy2008}). 
 \begin{theorem}
 \label{th:conditions_SPSA_convergence}
The sequence $\{\theta_k\}_{k\geq 0}$ in (\ref{eq:model_update}) converges weakly to a locally optimal parameter $\theta^*$.
\end{theorem}
\begin{proof}
	See Appendix \ref{appen:conditions_SPSA_convergence}.
\end{proof}

\section{Stochastic Optimization Method: Partially Observed Graph Process with Known Dynamics}
\label{sec:stochastic_optimization_partially_observed_graph}
In this section, we assume that the influence maximizer  can observe only a small part of the social network graph at each time instant. The aim of this section is to combine the stochastic optimization framework proposed in Sec. \ref{sec:stochastic_optimization_perfectly_observed_graph} to this partially observed setting.

\subsection{Partially Observed Graph Process}
\label{subsec:partially_observed_graph_process}
In some applications, the influence maximizer can observe only a small part of the full network $G_n$, at any time instant $n$. Let $G^{\bar{V}}$ denote the subgraph of $G = (V,E)$, induced by the set of nodes $\bar{V} \subseteq V$. Then, we consider the case where, the observable part is the subgraph of $G_n$ which is induced by a fixed subset of nodes $\bar{V} \subseteq V$ i.e.  the observable part at time $n$ is $G_n^{\bar{V}}$ ($G^{\bar{V}}$ denotes the subgraph of $G$, induced by the set of nodes $\bar{V}$)\footnote{For example consider the friendship network $G_n= (V, E_n)$ of all the high school students in a city at time $n$. The smaller observable part could be the friendship network formed by the set of students in a particular high school $V'  \subset V$, which is a subgraph of the friendship network $G_n$. The influence maximizer then needs to perform influence maximization by observing this subgraph.}. Then, the observation space $\bar{\mathcal{G}}$ of the  Markovian graph process can be defined as, 
\begin{equation}
\bar{\mathcal{G}} = \bigcup_{G \in \mathcal{G}} G^{\bar{V}},
\end{equation}
which  consists of the subgraphs induced by $\bar{V}$ in each graph $G \in \mathcal{G}$ (the case ${\bar{V} = V}$ corresponds to the perfectly observed case). For each $G \in \mathcal{G}$  and $\bar{G} \in \bar{\mathcal{G}}$, the observation likelihoods, denoted by $B_{G\bar{G}}$  are defined as,
\begin{align}
B_{G\bar{G}} &= \mathbb{P}(G_n^{\bar{V}} = \bar{G} \,\bigr\vert\, G_n = G).
\end{align}
In our system model, these observation likelihoods can take only binary values: 
\begin{equation}
  B_{G\bar{G}} =
  \begin{cases}
    1 & \text{if $G^{\bar{V}} = \bar{G}$} \\
    0 & \text{otherwise}
  \end{cases}.
\end{equation}
i.e. $B_{G\bar{G}} = 1$ if the subgraph formed by the set of nodes $\bar{V}$ in graph $G \in \mathcal{G}$ is $\bar{G} \in \mathcal{Y}$ and, $0$ otherwise. In this setting, our main assumption is the following. 
\begin{assumption}
	\label{assumption:partially_observed_graph}
	The measurement likelihood matrix $B$ and the parameterized transition probability matrix $\parmat$, are known to the influence maximizer but, the (finite) sample paths of the Markovian graph process are observed in noise.  
\end{assumption}

Then, the aim of the influence maximizer is to recursively solve the problem of influence maximization given in (\ref{eq:optimization_problem}), utilizing the information assumed to be known in Assumption \ref{assumption:partially_observed_graph}. However, solving (\ref{eq:optimization_problem}) is non-trivial even under the Assumption \ref{assumption:partially_observed_graph}, as emphasized in the following remark. 

\begin{remark}
	\normalfont
	 Even when $\parmat$ is known, it is intractable to compute $\parstatdist$ in closed form for many cases \cite{krishnamurthy2016}. Hence, assumption \ref{assumption:partially_observed_graph} does not imply that a closed form expression of the objective function in (\ref{eq:influence_1}) can be obtained and hence, solving 	(\ref{eq:optimization_problem}) still remains a non-trivial problem. 
\end{remark}

\subsection{Randomized Influence Maximization using HMM Filter Estimates}
\label{subsec:hmm_filter}

\begin{algorithm}
	\caption{Hidden Markov Model (HMM) Filter Algorithm for Tracking the Graph Process}
	\label{alg:hmm_filter}
	\DontPrintSemicolon 
	\KwIn{$\parmat, B$ and intial prior $\hmmprior$}
	\KwOut{Finite sample estimate $ \hat{C}^\theta_N$ of the influence function $C(\theta)$}
	
	\vspace{0.3cm}
	
	\begin{enumerate}
		\item For every time instant $n = 1, 2, \dots, N-1$, given observation $G^{\bar{V}}_{n+1}$, update the $\vert \mathcal{G}\vert$-dimensional posterior:
		\begin{equation}
		\label{eq:hmm_recursion}
			\pi_{n+1}^{\theta} = T(\pi_n,G^{\bar{V}}_{n+1}) = \frac{B_{G^{\bar{V}}_{n+1}}P'_\theta\hmmpost}				{\sigma(\hmmpost,G^{\bar{V}}_{n+1})}
		\end{equation} where,
		\begin{equation}
			\sigma(\hmmpost,G^{\bar{V}}_{n+1}) = \mathbf{1}'B_{G^{\bar{V}}_{n+1}}P'_\theta\hmmpost
		\end{equation}
		and, $\mathbf{1}$ denotes the column vector with elements equal to one. 
		\item Compute the estimate of the influence function $C(\theta)$,
			\begin{equation}
			\label{eq:hmm_cost_est}
			 \hat{C}^\theta_N = \frac{\sum_{n = 0}^{N-1}\hat{c}_\theta'\hmmpost}{N}.
			\end{equation} where, $\hat{c}_\theta$ denotes the column vector with elements $\hat{c}(\theta, G_i), \, G_i\in \mathcal{G}$.
		
	\end{enumerate}
		
\end{algorithm}

Assumption \ref{assumption:partially_observed_graph} made in \ref{subsec:partially_observed_graph_process}  makes it possible to implement an HMM filter (see \cite{krishnamurthy2016} for a detailed treatment of HMM filters and related results). The HMM filter is a finite dimensional Bayesian filter which recursively (with each observation) computes $\hmmpost$ which is the probability distribution  of the state of the graph, conditional on the sequence of  observations $\{G^{\bar{V}}_{0},\,\dots, G^{\bar{V}}_{n}\}$. Algorithm \ref{alg:hmm_filter} gives the HMM filter algorithm and, Theorem \ref{th:asymptotic_unbiasedness} establishes the asymptotic convergence of the influence function estimate obtained from from it.  

\begin{theorem}
\label{th:asymptotic_unbiasedness}
The finite sample estimate $\hat{C}^\theta_N$  of the influence function obtained in (\ref{eq:hmm_cost_est}) is an asymptotically unbiased estimate of the influence function $C(\theta)$
i.e. 
\begin{equation}
\lim_{N\to\infty} \mathbb{E}\{\hat{C}^\theta_N\} = C(\theta).
\end{equation}
Further, $\{\theta_k\}_{k\geq 0}$ in (\ref{eq:model_update}) converges weakly to a locally optimal parameter $\theta^*$, when estimates computed in steps 4 and 6 of Algorithm \ref{alg:stochastic_optimization} are replaced by the estimates obtained using the using Algorithm \ref{alg:hmm_filter}. 

\end{theorem}
\begin{proof}
See Appendix \ref{appen:asymptotic_unbiasedness}
\end{proof}

Finally, the Algorithm \ref{alg:stochastic_optimization} can be modified by replacing the estimates computed in steps 4 and 6 with the influence function estimates obtained using the HMM filter in Algorithm \ref{alg:hmm_filter}. 

\section{Numerical Results}
\label{sec:numerical_results}

\begin{figure*}[!htb]
	\centering
	\begin{subfigure}[!h]{0.5\textwidth}
		\centering
		\includegraphics[width=3.5in]{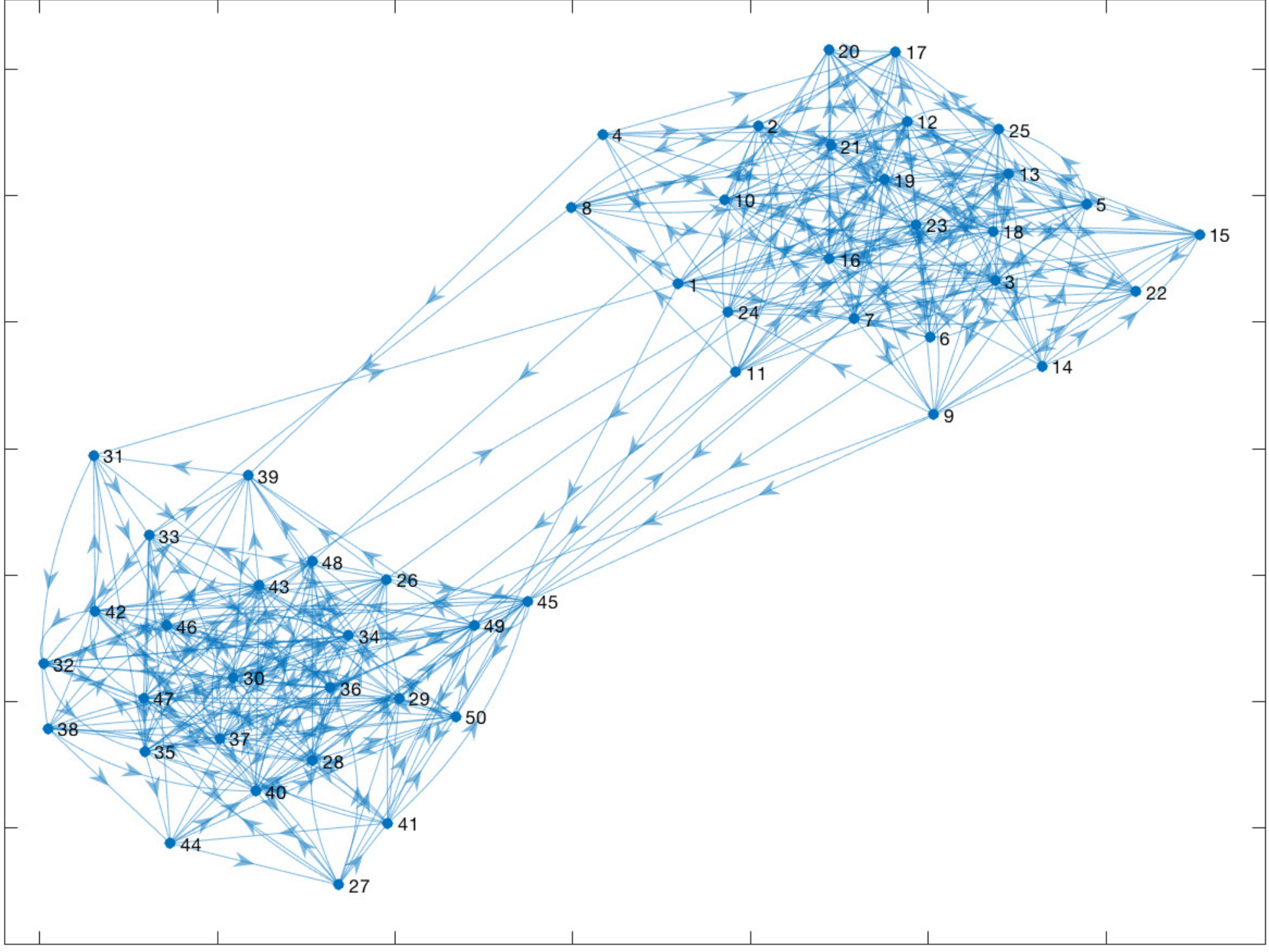}
		\caption{Graph $G^1 = (V, E^1)$ with two equal sized clusters }
		\label{subfig:G1}
	\end{subfigure}\hfill
	\begin{subfigure}[!h]{0.5\textwidth}
		\centering
		\includegraphics[width=3.5in]{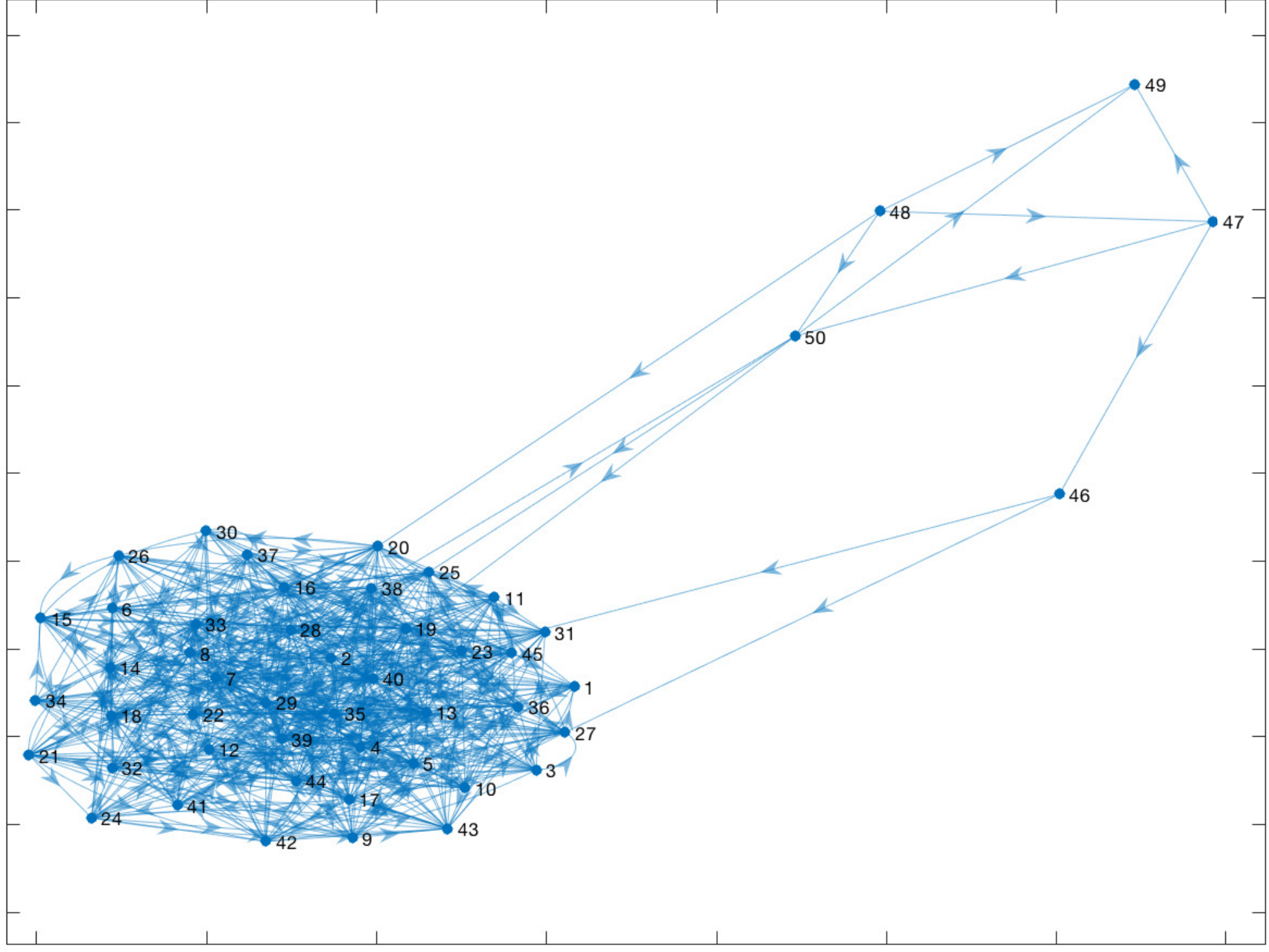}
		\caption{Graph $G^2 = (V, E^2)$ with a single large cluster}
		\label{subfig:G2}
	\end{subfigure}
	\caption{State space $\mathcal{G}$, of the Markovian graph process $\{G_n\}_{n\geq 0}$}
	\label{fig:states_of_graph_process}
\end{figure*}

In this section, we apply the stochastic optimization algorithm presented in Sec. \ref{sec:stochastic_optimization_perfectly_observed_graph} to an example setting and, illustrate its convergence with a feasible number of iterations. 

\subsection{Experimental Setup}
\label{subsec:experimental_setup}

We use the Stochastic Block Model (SBM) as a generative models to create the graphs used in this section. These models have been widely studied in statistics \cite{zhao2012consistency, abbe2016exact, rohe2011spectral, lelarge2015reconstruction, fishkind2013consistent} and network science \cite{karrer2011stochastic, wilder2017influence} as  generative models that closely resemble the real world networks. 
 
\vspace{0.1cm}
{\bf State space of the graph process: } We consider the  graph process obtained by Markovian switching between the two graphs in Fig. \ref{fig:states_of_graph_process}: a graph where two dense  equal sized clusters exist (Graph $G^1$) and, a graph where most of the nodes (45 out of 50) are in a single dense  cluster (Graph $G^2$).  These graphs are sampled from SBM models with the following parameter values: $G^1$  with cluster sizes 25, 25 with, within cluster edge probability $p_w^{SBM} = 0.3$, between cluster edge probability $p_b^{SBM} = 0.01$ and, $G^2$ with cluster sizes 45, 5 with, within cluster edge probability $p_w^{SBM} = 0.3$, between cluster edge probability $p_b^{SBM} = 0.01$. This graph process is motivated by the clustered and non-clustered states of a social network.

\vspace{0.1cm}
{\bf Sampling distribution and the graph evolution:} We consider the case where the influence maximizer samples from a subset of nodes that consists of the two nodes indexed by $v_1$ and $v_2$ using the parameterized probability distribution ${\pardist = [\cos^2(\theta)\quad \sin^2(\theta)]}'$. These two nodes are located in different clusters in the graphs $G^1, G^2$. Also, the transition probabilities may depend on this sampling distribution (representing for example, the adversarial networks/nodes as explained in Sec. \ref{subsec:practical_example}). However, exact functional characterizations of such dependencies are not known in many practical applications. Also, the form of these dependencies may change over time as well (recall Remark \ref{remark:functional_dependency}).  This experimental setup considers a graph process with a stationary distribution $\parstatdist = \pardist$ in order to have a closed form influence function as a ground truth (in order to compare the accuracy of the estimates). In an actual implementation of the algorithm, this functional dependency need not be known to the influence maximizer. 

\vspace{0.1cm}
{\bf Influence Functions: } The transmission time distribution was selected to be an exponential distribution with mean $1$ for edges within a cluster and, an exponential distribution with mean $10$ for between cluster edges (in the SBM). Further, the message expiration time was selected to be $T =  1.5$. Then, the influences of nodes $v_1, v_2$ on graphs $G_1$ and $G_2$ were estimated to be as follows by evaluating the integral in Definition \ref{defn:influence_on_G} with a naive simulation method (repeated use of the shortest path algorithm): $\sigma_{G^1}(v_1,T) = 25.2,\, \sigma_{G^1}(v_2,T) = 23.2,\,\sigma_{G^2}(v_1,T) = 45.1,\, \sigma_{G^2}(v_2,T) = 5.	8$. These values, along with the expressions for $p_\theta$ and $\parstatdist$ were used to obtain the following expression for the influence function (defined in Definition \ref{defn:influence_function}) to be compared with the outputs of the algorithm estimates:
\begin{align}
C(\theta) &= \parstatdist(G^1)(\sigma_{G^1}(v_1,T)p_\theta(v_1) + \sigma_{G^1}(v_2,T)p_\theta(v_2))\nonumber\\ &\hspace{0.75cm}+\parstatdist(G^2)(\sigma_{G^2}(v_1,T)p_\theta(v_1)+\sigma_{G^2}(v_2,T)p_\theta(v_2) )\nonumber\\
&= 17.9\sin^2(\theta) - 37.3\sin^4(\theta) + 25.2.
\end{align} In this context, the goal of our algorithm is to locate the value of $\theta$ which maximizes this function, without using knowledge of $C(\theta)$ or $\parstatdist$.

\subsection{Convergence of the Recursive Algorithm for Influence Maximization}
\begin{figure}[!h]
	\centering
	\includegraphics[width=\linewidth]{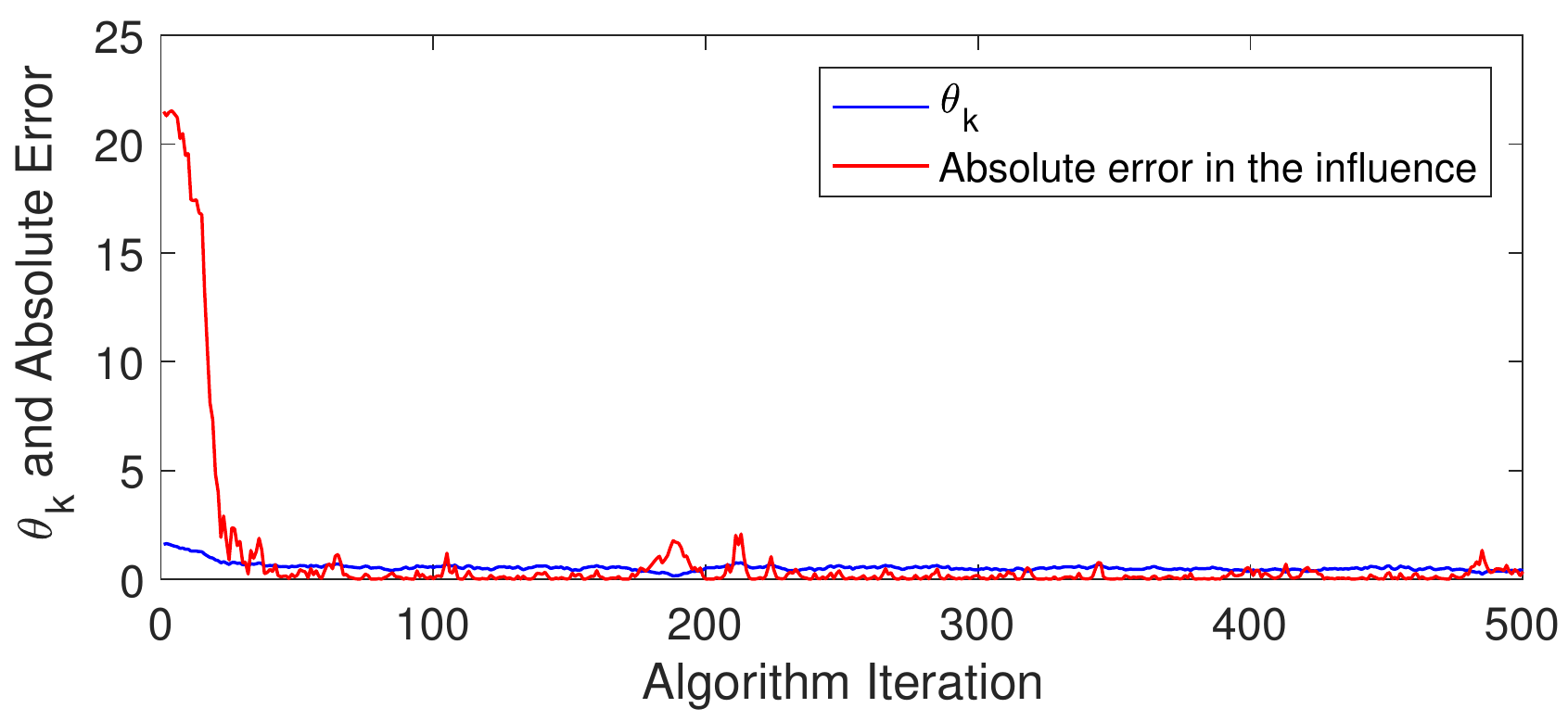}
	\caption{Parameter value $\theta_k$ and the absolute value of the error (difference of current and maximum influence) versus algorithm iteration, showing the convergence of the algorithm in a feasible number of iterations. }
	\label{fig:convergence}
\end{figure}

Algorithm \ref{alg:stochastic_optimization} was utilized in an experimental setting with the parameters specified in the Sec. \ref{subsec:experimental_setup}. For this, the length of the observable sample path of the graph process was assumed to be $\bar{N} = 30$. Further, in the Algorithm \ref{alg:instantaneous_cost_estimation}, the number of transmission time sets ($s$) and the number of exponential random variable sets ($m$) were both set to be $10$. 

With these numerical values for the parameters, Fig. \ref{fig:convergence} shows the variation of the absolute error (absolute value of the difference between the current and maximum influence) and the parameter value, against the iteration of the algorithm.  From this, it can be seen that the algorithm finds optimal parameter in  less than 50 iterations. Further, any change in the system model will result in a suboptimal (expected) influence only for 50 iterations since the algorithm is capable of tracking the time evolving optima. Hence, this shows that the stochastic optimization algorithm is capable of estimating the optimal parameter value with a smaller (less than 50) number of iterations.

\subsection{Effect of variance reduction in convergence and tracking the optima of a time-varying influence function}
\begin{figure}[!h]
\centering
\includegraphics[width=\linewidth]{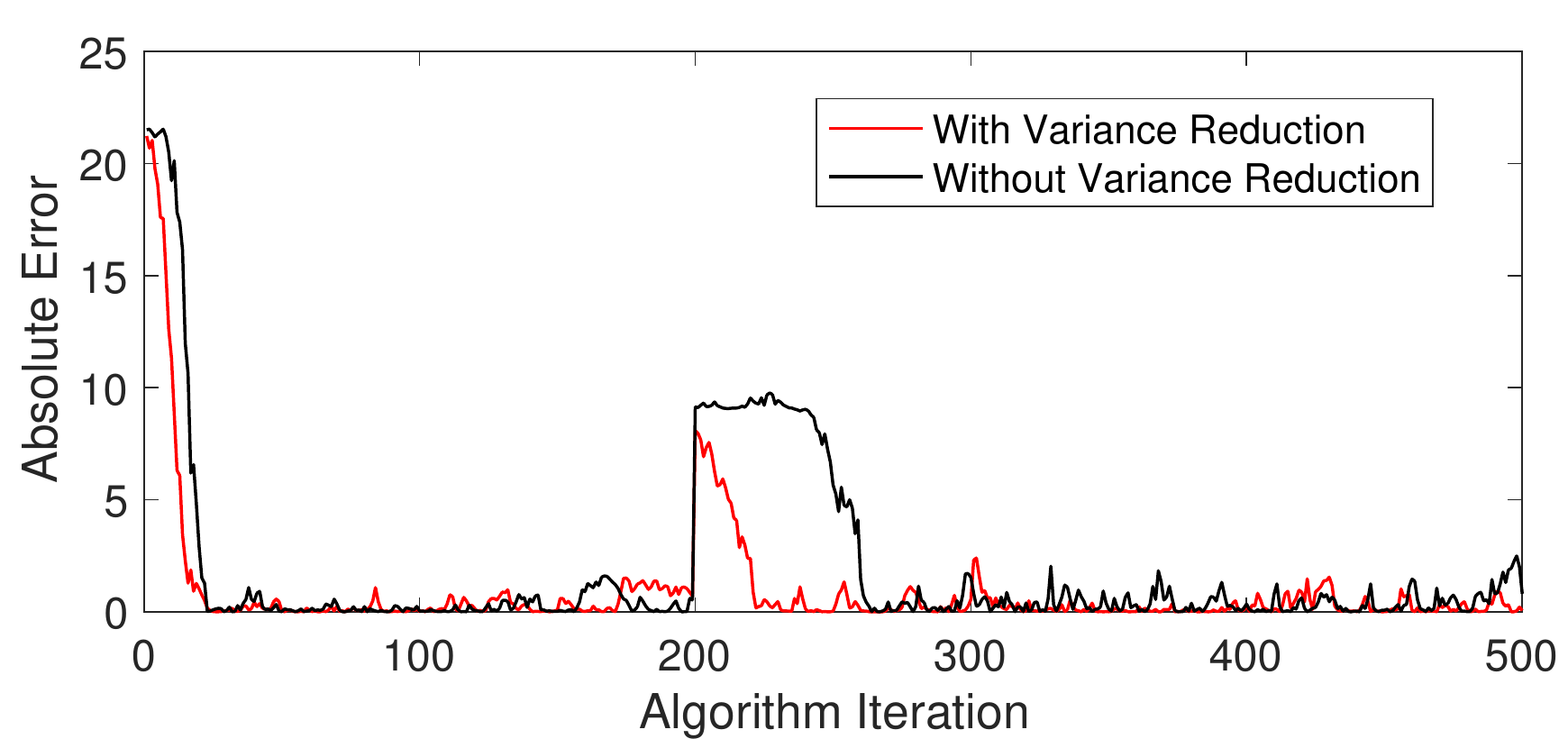}
\caption{Absolute value of the error versus algorithm iteration with and without the common random number based variance reduction. This depicts the importance of reduced variance in tracking the optima of a time evolving influence function.}
\label{fig:tracking_difference}
\end{figure}

Here we aim to see how the proposed stochastic approximation algorithm can track the optima when the system model changes on a slower time scale and, the effect of reduced variance Algorithm \ref{alg:instantaneous_cost_estimation} on the accuracy of tracking. For this, the experimental setup described in \ref{subsec:experimental_setup} was utilized again and, a sudden change in the influence function (by changing the the state space of the graph process and the functional dependency) was introduced at the iteration number 200. In this setting, Fig. \ref{fig:tracking_difference} depicts the variation of the absolute error in influence with the algorithm iteration for two cases: with the reduced variance Algorithm \ref{alg:instantaneous_cost_estimation} (red curve) and without the variance reduction approach (black curve). It can be seen from Fig. \ref{fig:tracking_difference} that the variance reduction method improves the speed of convergence to the optima initially (iterations 1 to 50) and, also in tracking the optima after a sudden change in the influence function (iterations 200 to 300). Further, after convergence (between iterations 50 to 200 and 300 to 500), it can be seen that reduced variance approach is less noisy compared to the method without variance reduction. Hence, this shows that the proposed approach is capable of tracking the optimal sampling distribution in a slowly evolving system such as, varying graph state space, evolving functional dependencies, etc.

\section{Conclusion}
\label{sec:conclusion}
This paper considered the problem of randomized influence maximization over a Markovian Graph Process: given a fixed set of nodes whose connectivity graph is evolving as a Markov chain, estimate the probability distribution (over this fixed set of nodes) that samples a node which will initiate the largest information cascade (in expectation). The evolution of the graph was allowed to functionally depend on the sampling  probability distribution in order to keep the problem more general. This was formulated as a problem of tracking  the optimal solution of a (time-varying) optimization problem where, a closed form expression of the objective function (influence function) is not available. In this setting, two stochastic gradient algorithms were presented to estimate the optimal sampling distribution for two cases: 1) transition probabilities of the graph are unknown but, the graph can be observed perfectly 2) transition probabilities of the graph are known but, the graph is observed in noise. These algorithms are based on the Simultaneous Perturbation Stochastic Approximation Algorithm that requires only the noisy estimates of the influence function. These noisy estimates of the influence function were obtained by combining a neighborhood size estimation algorithm with a variance reduction method and then, averaging over a finite sample path of the graph process. The convergence of the proposed methods were established theoretically and, illustrated with numerical examples. The numerical results show that, with the reduced variance approach, the algorithms are capable of tracking the optimal influence in a time varying system model (with changing graph state spaces, etc.).


%

\appendices

\section{Proof of Theorem \ref{th:unbiasedness}}
\label{appen:unbiasedness}
Let the size of the $T-$distance neighborhood $\mathcal{N}_G(v,T) $ of a node $v \in V$ of graph $G = (V,E)$ conditional on a transmission time set $L_u = {\{\tau_{ji}\}^u_{(ji) \in E}}$ be denoted by $h^v(L_u)$. Further, Let $\hat{h}^v(L_u)$ denote $\frac{m-1}{\sum_{j= 1}^{m}\bar{r}^{v}_{u,j}}$. 
\begin{align}
\mathbb{E}\{X(v,G)\} &= \mathbb{E}\bigg\{\frac{1}{s} \sum_{u = 1}^{s} \frac{m-1}{\sum_{j= 1}^{m}\bar{r}^{v}_{u,j}}\bigg\}\\
									&= \frac{1}{s} \sum_{u = 1}^{s} \mathbb{E}\{{\hat{h}^v(L_u)}\}\\
									&= \frac{1}{s} \sum_{u = 1}^{s} \mathbb{E}\{\mathbb{E}\{\hat{h}^v(L_u)\vert L_u\}\}\\	
									&= \frac{1}{s} \sum_{u = 1}^{s} \mathbb{E}\big\{ h^v(L_u)\big\}\\	
									&\text{ (from conditional unbiasedness proved in \cite{cohen1997})}\nonumber\\	
									&= \sigma_G(v,T) \text{ (by (\ref{eq:influence_of_a_node_on_a_graph}))}
\end{align} 

To analyze the variance of $X(v,G)$, first note that $h^v(\{\tau_{ji}\}^u_{(ji) \in E})$ is a monotonically decreasing function of all its elements $\{\tau_{ji}\}^u_{(ji) \in E}$. Further, the following result about monotone functions of random variables from \cite{ross2013simulation} will be used to establish the result. 
\begin{lemma}
	\label{lemma:negative_cor}
	If $g(x_1, ... , x_n)$ is a monotone function of each of its arguments, then, for a set $U_1 , ... , U_n$ of independent random numbers.
	\begin{equation}
	\cov(g(U_1 , ... , U_n),g(1-U_1 , ... , 1-U_n)) \leq 0.
	\end{equation}
\end{lemma}

 Now, consider the variance of $\frac{\hat{h}^v(L_u) + h(L_{s/2+u})}{2}$ where $L_u$ and $L_{s/2+u}$ are the pair of correlated transmission time sets as defined in (\ref{eq:correlated_transmission_time_sets_1}) and (\ref{eq:correlated_transmission_time_sets_2}).
\begin{multline}
\var\bigg(\frac{\hat{h}^v(L_u) + \hat{h}(L_{s/2+u})}{2}\bigg) = \\
\frac{1}{4}\bigg(\var\big(\hat{h}^v(L_u)\big) + \var\big(\hat{h}^v(L_{s/2+u})\big) \\+  2\cov\big(\hat{h}^v(L_u),\hat{h}^v(L_{s/2+u})\big)\bigg)
\end{multline}
\begin{align}
&=\frac{1}{2}\bigg(\var\big(\hat{h}^v(L_u)\big) +  \cov\big(\hat{h}^v(L_u),\hat{h}^v(L_{s/2+u})\big)\bigg)\label{eq:variance_correlated}\\
&\text{ (since $\hat{h}^v(L_u)$ and $\hat{h}^v(L_{s/2+u})$ are identically distributed)}\nonumber	
\end{align}
Now consider $\cov\big(\hat{h}^v(L_u),\hat{h}^v(L_{s/2+u})\big)$. By using the law of total covariance,
\begin{align}
&\cov\big(\hat{h}^v(L_u),\hat{h}^v(L_{s/2+u})\big) =\nonumber\\
&\hspace{0.25cm}\mathbb{E}\big\{\cov\big(\hat{h}^v(L_u),\hat{h}^v(L_{s/2+u})\big\vert  \{U^u_{ji}\}_{(ji) \in E}  \big)\big\} + \nonumber\\
&\hspace{0.25cm}\cov\big(\mathbb{E}\{\hat{h}^v(L_u)\vert \{U^u_{ji}\}_{(ji) \in E}\},\mathbb{E}\{\hat{h}^v(L_{s/2+u})\vert \{U^u_{ji}\}_{(ji) \in E}\}  \big)\\
&\hspace{0.25cm}= \cov\big(\mathbb{E}\{\hat{h}^v(L_u)\vert \{U^u_{ji}\}_{(ji) \in E}\},\mathbb{E}\{\hat{h}^v(L_{s/2+u})\vert \{U^u_{ji}\}_{(ji) \in E}\}\  \big)\\
&\text{(since $\hat{h}^v(L_u)$ and $\hat{h}^v(L_{s/2+u})$ are uncorrelated given  $\{U^u_{ji}\}_{(ji) \in E}$}\nonumber\\
&\hspace{0.25cm}= \cov\big({h}^v(L_u),{h}^v(L_{s/2+u})  \big)\\
&\text{(from the conditional unbiasedness proved in \cite{cohen1997})}\nonumber\\
&\hspace{0.25cm}= \cov\Big({h}^v(\{F_{tt}^{-1}(U^u_{ji})\}_{(ji) \in E}),{h}^v(\{F_{tt}^{-1}(1 - U^u_{ji})\}_{(ji) \in E}) \Big)
\end{align}
$F_{tt}^{-1}(\cdot)$ is a monotone function (inverse of a CDF) and, $h^v(\{\tau_{ji}\}^u_{(ji) \in E})$ is also monotone in all its arguments $\{\tau_{ji}\}^u_{(ji) \in E}$.  Hence, the composite function ${h}^v(\{F_{tt}^{-1}(U^u_{ji})\}_{(ji) \in E})$ is monotone is all its arguments $\{U^u_{ji}\}_{(ji) \in E}$ (because, the composition of monotone functions is monotone). Then, from Lemma \ref{lemma:negative_cor}, it follows that 
\begin{equation}
\cov\big(\hat{h}^v(L_u),\hat{h}^v(L_{s/2+u})\big) \leq 0.
\end{equation}
Then, from (\ref{eq:variance_correlated}), it follows that,
\begin{equation}
\label{eq:variance_reduction}
\var\bigg(\frac{\hat{h}^v(L_u) + \hat{h}(L_{s/2+u})}{2}\bigg) \leq \frac{1}{2}\var\big(\hat{h}^v(L_u).
\end{equation}
Then, by applying the total variance formula to the left hand side of (\ref{eq:variance_reduction}) and, using the fact $\var\{\hat{h}^v(L_u)\vert L_u \} = \frac{h^2(L_u)}{m-2}$(from \cite{cohen1997}), we get,
\begin{multline}
\var\bigg(\frac{\hat{h}^v(L_u) + \hat{h}(L_{s/2+u})}{2}\bigg) \leq \\ \frac{1}{2}\bigg({\frac{\sigma_G(v,T)^2}{m-2} +  \frac{(m-1)\var(\vert\mathcal{N}_G(v,T)\vert)}{m-2}}\bigg)
\end{multline} and, the proof follows by noting that $X(v,G)$ is the average of $\frac{\hat{h}^v(L_u) + h(L_{s/2+u})}{2}$ for $u = 1,\cdots, s/2$. 

Proof of Part \ref{th:unbiasedness_part2} of Theorem \ref{th:unbiasedness} follows from similar arguments as above and hence omitted. 

\section{Proof of Theorem  \ref{th:conditions_SPSA_convergence}}
\label{appen:conditions_SPSA_convergence}
The following result from \cite{kushner2003} will be used to establish the weak convergence of the sequence $\{\theta_k\}_{k \geq 0}$ obtained in Algorithm \ref{alg:stochastic_optimization}.

Consider the stochastic approximation algorithm,
\begin{equation}
\label{eq:stochastic_approximation_proof}
\theta_{k+1} = \theta_{k} + \epsilon H(\theta_k,x_k), \quad k = 0,1,\dots
\end{equation} where $\epsilon > 0$, $\{x_k\}$ is a random process and, $\theta \in \mathbb{R}^p$ is the estimate generated at time $k = 0,1,\dots$.
Further, let 
\begin{equation}
\label{eq:interpolated_trajectory}
\theta^\epsilon(t) = \theta_k \text{ for } t \in [k\epsilon,k\epsilon + \epsilon], \quad k = 0,1,\dots,
\end{equation}
which is a piecewise constant interpolation of $\{\theta_k\}$.
In this setting, the following result holds. 

\begin{theorem}
	\label{th:SA}
	Consider the stochastic approximation algorithm (\ref{eq:stochastic_approximation_proof}). Assume
	\begin{itemize}
		\item[{\bf SA1:}] $H(\theta, x)$ us uniformly bounded for all $\theta \in \mathbb{R}^p$ and $x \in \mathbb{R}^q$.
		\item[{\bf SA2:}] For any $l \geq 0$, there exists $h(\theta)$ such that
		\begin{equation}
		\frac{1}{N} \sum_{k = l}^{N+l-1} \mathbb{E}_l\{H(\theta, x_k)\} \rightarrow h(\theta) \text{ as } N \rightarrow \infty.
		\end{equation} where, $\mathbb{E}_l \{\cdot\}$ denotes expectation with respect to the sigma algebra generated by $\{x_k: k < l\}$.\label{th:SA_condition_2}
 		
		\item[{\bf SA3:}] The ordinary differential equation (ODE)
		\begin{equation}
		\label{eq:ODE}
		\frac{d\theta(t)}{dt} = h(\theta(t)),\quad  \theta(0) = \theta_0
		\end{equation} has a unique solution for every initial condition. \label{th:SA_condition_3}
	\end{itemize}

Then, the interpolated estimates  $\theta^\epsilon(t)$ defined in (\ref{eq:interpolated_trajectory}) satisfies
\begin{equation}
\lim_{\epsilon \to 0} \mathbb{P}\big( \sup_{0 \leq t \leq T}  \vert \theta^\epsilon(t) - \theta(t)  \vert \geq \eta \big) = 0 \text{ for all } T>0, \eta > 0
\end{equation} where, $\theta(t)$ is the solution of the ODE (\ref{eq:ODE}).
\end{theorem}
The condition {SA1} in Theorem \ref{th:SA} can be replaced by uniform integrability and the result still holds \cite{krishnamurthy2016}. 

Next, we show how Algorithm \ref{alg:stochastic_optimization} fulfills the assumptions SA1, SA2, SA3 in Theorem \ref{th:SA}. Detailed steps of similar proofs related to stochastic approximation algorithms can be found in \cite{krishnamurthy2002recursive} and Chapter 17 of \cite{krishnamurthy2016}.

Consider $\sup_{k}  \mathbb{E} \vert\vert \hat{\nabla}{C_k(\theta_{k})}  \vert\vert $ defined in Algorithm \ref{alg:stochastic_optimization}.
\begin{align}
\sup_{k}  \mathbb{E} \vert\vert \hat{\nabla}{C_k(\theta_{k})}  \vert\vert  &= \sup_{k}  \mathbb{E}\bigg|\bigg|\bigg(\frac{1}{\bar{N}}\sum_{n = 2kN}^{2kN +\bar{N}-1} \hat{c}(\theta,G_n) \\ 
&\hspace{1cm}- \frac{1}{\bar{N}}\sum_{n = (2k+1)N}^{(2k+1)N +\bar{N}-1} \hat{c}(\theta,G_n)\bigg)\frac{d_k}{2\Delta} \bigg|\bigg| \nonumber\\
&=\sup_{k} \frac{\sqrt{M}}{2\Delta\bar{N}}  \mathbb{E} \bigg\vert \sum_{n = 2kN}^{2kN +\bar{N}-1} \hat{c}(\theta,G_n) \\ 
&\hspace{1cm}- \sum_{n = (2k+1)N}^{(2k+1)N +\bar{N}-1} \hat{c}(\theta,G_n) \bigg\vert \\
&\leq \sup_{k}  \frac{\sqrt{M}}{2\Delta\bar{N}}   \mathbb{E} \bigg\{ \sum_{n = 2kN}^{{(2k+1)N +\bar{N}-1}}\big| \hat{c}(\theta,G_n) \big| \bigg\}\\ 
&\hspace{1cm} \text{(By triangle inequality)}\nonumber\\
&= \sup_{k}  \frac{\sqrt{M}}{2\Delta\bar{N}}  \sum_{n = 2kN}^{{(2k+1)N +\bar{N}-1}} \mathbb{E}  \{\hat{c}(\theta,G_n)\} \\ 
&\hspace{1cm} \text{(Since  $\hat{c}(\theta,G_n) \geq 0$)} \nonumber\\
&= \sup_{k}  \frac{\sqrt{M}}{2\Delta\bar{N}}  \sum_{n = 2kN}^{{(2k+1)N +\bar{N}-1}} \mathbb{E}  \{{c}(\theta,G_n)\} \\ 
&\hspace{1cm}\hspace{0cm} \text{(Conditioning on $G_n$ and,} \nonumber\\
&\hspace{1.5cm}\text{using Part \ref{th:unbiasedness_part2} of Theorem \ref{th:unbiasedness})} \nonumber
\end{align}
\begin{align}
 &\leq \sup_{k}  \frac{\sqrt{M}}{2\Delta\bar{N}}  \sum_{n = 2kN}^{{(2k+1)N +\bar{N}-1}} \big\{\max_{v\in V, G\in \mathcal{G}} \sigma_G(v,T)\big\} \\
 &=\frac{\sqrt{M}}{\Delta}  \max_{v\in V, G\in \mathcal{G}} \sigma_G(v,T)\\
 &\hspace{0cm} \text{(maximum exists since $V, \mathcal{G}$ are finite sets)}\nonumber
\end{align}
Hence the uniform integrability condition (alternative for SA1) is fulfilled.

\vspace{0.25cm}
Next, note that ${C_k(\theta_{k})}$ is an asymptotically (as $N$ tends to infinity) unbiased estimate of $C(\theta_k)$ by the uniform integrability and almost sure convergence (by law of large numbers for ergodic Markov chains). Therefore, as $\Delta$ (perturbation size in (\ref{eq:gradient_estimate})) tends to zero, $\hat{\nabla}{C_k(\theta_{k})}$ in Algorithm \ref{alg:stochastic_optimization} fulfills the SA2 condition. 

\vspace{0.25cm}
SA3 is fulfilled by the (global) Lipschitz continuity of the gradient $\nabla_\theta C(\theta)$ which is a sufficient condition for the existence of a unique solution for a non-linear ODE (for any initial condition) \cite{khalil1996noninear}.

\section{Proof of Theorem \ref{th:asymptotic_unbiasedness}}
Consider the expected value of the finite sample estimate of the influence function (output of Algorithm \ref{alg:hmm_filter})
\label{appen:asymptotic_unbiasedness}
\begin{align}
\mathbb{E}\{\hat{C}^\theta_N\} &= \mathbb{E}\{\frac{\sum_{n = 1}^{N}\hat{c}^T\hmmpost}{N}\}\\
& = \frac{c^T}{N} \sum_{n = 1}^{N} \mathbb{E}\{\hmmpost\}\\ 
&\text{ (by Theorem \ref{th:unbiasedness}})\nonumber\\
& = \frac{c^T}{N} \sum_{n = 1}^{N} (\parmat^n)^T\pi_0 \label{eq:cessaro_sum}
\end{align}
Then, the result follows by noting that (\ref{eq:cessaro_sum}) is a Ces\`{a}ro summation of a convergent sequence. 

The weak convergence of the sequence $\{\theta_k\}_{k\geq 0}$ (with HMM filter estimates) follows again from Theorem \ref{th:SA}. Since the proof is mostly similar to the Proof of Theorem \ref{th:conditions_SPSA_convergence}, we skip the proof. 
%
%

\ifCLASSOPTIONcaptionsoff
  \newpage
\fi



\bibliographystyle{ieeetr}
\end{document}